\documentclass[11pt]{article}

\usepackage{float}
\usepackage{algorithm,tikz,amsmath,amssymb,amsthm,cite,fullpage,hyperref}
\usepackage[noend]{algpseudocode}

\usepackage{titling}
\thanksmarkseries{arabic}

\newtheorem{lemma}{Lemma}

\newtheorem{theorem}{Theorem}
\newtheorem{claim}{Claim}

\newcommand{\mA}{\mathcal{A}}
\newcommand{\XX}{\mathcal{X}}
\newcommand{\DD}{\mathcal{D}}
\newcommand{\OO}{\mathcal{O}}
\newcommand{\TT}{\mathcal{T}}
\newcommand{\GG}{\mathcal{G}}
\newcommand{\EE}{\mathcal{E}}

\newcommand{\probl}[3]{
\begin{flushleft}
\fbox{
\begin{minipage}{16cm}
\noindent {\sc #1}\\
          {\bf Input:} #2\\
          {\bf Output:} #3
\end{minipage}}
\medskip
\end{flushleft}
}
\newcommand{\paraprobl}[4]
{
  \begin{flushleft}
    \fbox{
      \begin{minipage}{16cm}
        \noindent {\textsc {#1}}\\
        {\bf Input:} #2\\
        {\bf Parameter:} #4\\
        {\bf Question:} #3
      \end{minipage}
    }
  \end{flushleft}
}

\newcommand {\np}       {{\sf NP}}

\newcommand {\nph}      {{\sf NP}\textrm{-hard}}
\newcommand {\fpt}      {{\sf FPT}}
\newcommand {\xp}       {{\sf XP}}
\newcommand {\wone}     {{\sf W}[1]}
\newcommand {\wonehard} {{\sf W}[1]\textrm{-hard}}

\newcommand {\nohyp}    {{\sf P}=\np}
\newcommand {\pw}       {{\sf pw}}
\newcommand {\tw}       {{\sf tw}}
\newcommand {\runningtitle}[1] {\vspace{0.5ex}\noindent{\textbf{\boldmath #1:}}}

\newcommand {\norm}[1]{\left\Vert#1\right\Vert}
\newcommand {\abs}[1]{\left\vert#1\right\vert}
\newcommand {\set}[1]{\left\{#1\right\}}
\newcommand {\defined} {\stackrel{def} {=}}

\newcommand{\sm}{\setminus}
\newcommand{\es}{\emptyset}

\newcommand {\OneThreeSAT} {{\sc Monotone 1 in 3-SAT }}
\newcommand {\mcc} {\textsc{Multicolored Clique}}
\newcommand {\mrcf}[1] {{\sc $#1$-MRCF}}
\newcommand {\rmrcf} {{\mrcf{r}}}
\newcommand {\rc} {{\sf rc}}
\newcommand {\tc} {{\sf tc}}
\newcommand {\inst} {(G,\chi,c,k)}
\newcommand {\instprime} {(G',\chi',c',k')}
\newcommand {\kmin} {{k_{\rm min}}}
\newcommand{\vect}[1]{\mathbf{#1}}
\newcommand{\zero}{\vect{0}}

\newcommand {\avgdeg} {\overline{{\sf deg}}(G)}
\newcommand{\commentfig}[1]{#1}

\title{{Minimum Reload Cost Graph Factors  Graphs\thanks{\href{http://www.cs.upc.edu/~sedthilk/oapf/apexpse_foto.jpg}{}Emails:\! {\scriptsize  \sf \href{mailto:julien.baste@lip6.fr}{julien.baste@lip6.fr},\!
\href{mailto:didem.gozupek@gtu.edu.tr}{didem.gozupek@gtu.edu.tr},\!
\href{mailto:cmshalom@telhai.ac.il}{cmshalom@telhai.ac.il},\!
\href{mailto:sedthilk@thilikos.info}{sedthilk@thilikos.info}.}
}}}

\author{Julien Baste%
\thanks{Sorbonne Université, Laboratoire d'Informatique de Paris 6, LIP6, Paris, France}%
\and 
Didem G\"{o}z\"{u}pek%
\thanks{Department of Computer Engineering, Gebze Technical University, Kocaeli, Turkey.}
\and 
Mordechai Shalom%
\thanks{TelHai College, Upper Galilee, 12210, Israel.}
\and
Dimitrios  M. Thilikos\thanks{AlGCo project-team, LIRMM, CNRS, Universit\'e de Montpellier, Montpellier, France.}$\ ^,$\thanks{Supported by projects DEMOGRAPH (ANR-16-CE40-0028) and ESIGMA (ANR-17-CE23-0010) and  by the Research Council of Norway and the French Ministry of Europe and Foreign Affairs, via the Franco-Norwegian project PHC AURORA 2019.}}

\hypersetup{
	pdftitle = {Minimum Reload Cost Graph Factors  Graphs},
	pdfauthor=  {Julien Baste, Didem G\"{o}z\"{u}pek, Mordechai Shalom, Dimitrios Thilikos},
	colorlinks = true,
	urlcolor = black!50!red,
	linkcolor = black!50!red,
	citecolor = black!50!green,
	menucolor = {red}
	filecolor = {cyan},
	anchorcolor = {black}
	urlcolor = {magenta},
	runcolor = {cyan},
	linkbordercolor = {blue}
}
\date{}

\begin{document}

\maketitle

\vspace{-1.3cm}
\begin{abstract}
\noindent The concept of \emph{Reload cost} in a graph refers to the cost that occurs while traversing a vertex via two of its incident edges. 
This cost is uniquely determined by the colors of the two edges. 
This concept has various applications in transportation networks, communication networks, and energy distribution networks. 
Various problems using this model are defined and studied in the literature.
The problem of finding a spanning tree whose diameter with respect to the reload costs is smallest possible,
the problems of finding a path, trail or walk with minimum total reload cost between two given vertices, 
problems about finding a proper edge coloring of a graph such that the total reload cost is minimized, 
the problem of finding a spanning tree such that the sum of the reload costs of all paths between all pairs of vertices is minimized,
and the problem of finding a set of cycles of minimum reload cost, that cover all the vertices of a graph, 
are examples of such problems.
In this work we focus on the last problem. 
Noting that a cycle cover of a graph is a 2-factor of it, we generalize the problem the the problem of finding an $r$-factor of minimum reload cost of an edge colored graph.
We prove several {\sf NP}-hardness results for special cases of the problem. 
Namely, bounded degree graphs, planar graphs, bounded total cost, and bounded number of distinct costs. 
For the special case of $r=2$, our results imply an improved {\sf NP}-hardness result. 
On the positive side, we present a polynomial-time solvable special case which provides a tight boundary between the polynomial and hard cases 
in terms of $r$ and the and the maximum degree of the graph. 
We then investigate the parameterized complexity of the problem, prove \wone-hardness results and present an \fpt-algorithm.
\end{abstract}
\medskip

\noindent{\bf keywords:} Parameterized Complexity, Graph Factors, Reload Costs

\section{Introduction}

Edge-colored graphs can be used to model optimization problems in diverse fields such as bioinformatics, communication networks, and transportation networks. 
\emph{Reload cost} in an edge-colored graph refers to the cost that occurs while traversing a vertex via two of its incident edges. 
This cost is uniquely determined by the colors of the two edges.

The reload cost concept has various applications in  transportation networks, communication networks, and energy distribution networks. For instance, a multi-modal cargo transportation network involves different means of transportation, where the (un)loading of cargo at transfer points is costly \cite{Ga08}. In energy distribution networks, transferring energy between different carriers cause energy losses and reload cost concept can be used to model this situation \cite{Ga08}. In communication networks, routing often requires switching between different technologies such as cable and fiber, where data conversion incurs high costs. Switching between different service providers in a communication network also causes switching costs \cite{Ga08}. Recently, \emph{dynamic spectrum access networks}, a.k.a. cognitive radio networks, received a lot of attention in the communication networks research community. Unlike other wireless networks, cognitive radio networks are envisioned to operate in a wide range of spectrum; therefore, frequency switching has adverse effects in delay and energy consumption \cite{AAK+13, GBA13, celik2016green}. This frequency switching cost depends on the frequency separation distance; hence, it corresponds to reload costs.

The reload cost concept was first introduced by Wirth and Steffan \cite{WiSt01} who focused on the problem of finding a spanning tree whose diameter with respect to the reload costs is smallest possible. 
Other works also focused on numerous optimization problems regarding reload costs: 
the problems of finding a path, trail or walk with minimum total reload cost between two given vertices \cite{GLMM10}, 
numerous path, tour, and flow problems \cite{AGM11}, 
the minimum changeover cost arborescence problem \cite{GGM11, GozupekSVZ14, gozupek2016constructing, GOP+16}, 
problems about finding a proper edge coloring of a graph such that the total reload cost is minimized \cite{gozupek2016edge}, 
and the problem of finding a spanning tree such that the sum of the reload costs of all paths between all pairs of vertices is minimized \cite{gamvros2012reload}.

An $r$-factor of a graph is an $r$-regular spanning subgraph. 
A $2$-factor is also called a \emph{cycle cover} and has many applications in areas such as computer graphics and computational geometry \cite{meijer2009algorithm}, for instance for fast rendering of 3D scenes. 
In an edge-weighted graph, the problem of finding a cycle cover with minimum cost can be solved in polynomial-time \cite{schrijver2003combinatorial}. 
Its reload cost counterpart was studied by Galbiati et.al. in \cite{GGM14}. 
In particular, they proved that the minimum reload cost cycle cover problem is $\nph$ even when the number of colors is $2$, the reload costs are symmetric and satisfy the triangle inequality. 
In this work, we build on this work by studying the minimum reload cost $r$-factor problem, which is a more generalized version of the minimum reload cost cycle cover problem. 
In particular, we prove several {\sf NP}-hardness results  for the special cases of this problem. 
Namely, bounded degree graphs, planar graphs, bounded total cost, and bounded number of distinct costs. 
For the special case of $r=2$, we prove an {\sf NP}-hardness result stronger than the one in \cite{GGM14}. 
On the positive side, we present a polynomial-time solvable special case. 
We then investigate the parameterized complexity of this problem, prove {\sf W}[1]-hardness results and present a fixed parameter tractable algorithm.

\section{Preliminaries}

\vspace{2mm}\runningtitle{Sets, vectors}
Given a non-negative integer $n$, we 
denote by $\mathbb{N}_{\geq n}$ the set of all integers $x$ such that $x \geq n$. 
If $n_{1},n_{2} \in \mathbb{N}_{\geq 0}$, we denote by  $[n_{1},n_{2}]$ 
the set of integers $x$ such that $n_1 \leq x \leq n_2$. 
We also use $[n]$ instead of $[1,n]$.
Given a finite set $X$ and an integer $s \in \mathbb{N}_{\geq 0}$, we denote by ${X\choose s}$ the set of all subsets of $A$ with exactly $s$ elements. 
For a set $X$ and an element $x$ we use $X + x$ and $X-x$ as shorthands for $X \cup \set{x}$ and $X \sm \set{x}$, respectively.
For two vectors $\vect{u}=(u_1,\ldots,u_d)$ and $\vect{v}=(v_1,\ldots,v_d)$ over the reals, we write $\vect{u} \leq \vect{v}$ if $u_i \leq v_i$ for every $i \in [d]$.
  
\vspace{2mm}\runningtitle{Graphs}   
All graphs we consider in this paper are undirected, finite, and without self loops or multiple edges. 
Given a graph $G$, we denote by $V(G)$ the set of vertices of $G$ and by $E(G)$ the set of edges of $G$. 
We say that a vertex $v\in V(G)$ and an edge $e\in E(G)$ are {\em incident} if $v\in e$, that is, $v$ is an endpoint of $e$.
Given a vertex $v\in V(G)$, we denote by $E_{G}(v)$ the set of edges of $G$ that are incident to $v$.
The {\em degree} of $v$ in $G$, denoted by ${\sf deg}_{G}(v)$ is $|E_{G}(v)|$. 
We also define the {\em maximum degree} of $G$ as  $\Delta(G)=\max \{{\sf deg}_{G}(v) \mid v \in V(G) \}$, the {\em minimum degree} of $G$ as  $\delta(G)=\min \{ {\sf deg}_{G}(v) \mid v \in V(G) \}$, and the average degree of $G$ as $\avgdeg$.
For a subset $X$ of $V(G)$, we denote by $G[X]$ the subgraph of $G$ induced by $X$.
A graph is {\em $r$-regular} if all its vertices have degree $r$.

We say that a graph $H$ is a {\em factor} of a graph $G$ when $V(H)=V(G)$ and $E(H)\subseteq E(G)$.
An $r$-regular factor of $G$ is termed an {\em $r$-factor} of $G$.

\vspace{2mm}\runningtitle{Parameterized complexity}
We refer the reader to~\cite{CyganFKLMPPS15,DF13} for basic background on parameterized complexity, and we recall here only some basic definitions.
A \emph{parameterized problem} is a language $L \subseteq \Sigma^* \times \mathbb{N}$.  For an instance $I=(x,k) \in \Sigma^* \times \mathbb{N}$, $k$ is called the \emph{parameter}.

A parameterized problem is \emph{fixed-parameter tractable} ({\fpt}) 
if there exists an algorithm $\mA$, 
a computable function $f$, and a constant $c$ such that given an instance $I=(x,k)$,
$\mA$ (called an {\fpt} \emph{algorithm}) correctly decides whether $I \in L$ in time bounded by $f(k) \cdot |I|^c$.

A parameterized problem is in {\xp} if there exists an algorithm $\mA$ and two computable functions $f$ and $g$ such that given an instance $I=(x,k)$,
$\mA$ (called an {\xp} \emph{algorithm}) correctly decides whether $I \in L$ in time bounded by $f(k) \cdot |I|^{g(k)}$. 

A parameterized problem with instances of the form $I=(x,k)$ is \emph{{\sf para}-\nph} if it is {\nph} for some fixed {\sl constant} value of the parameter $k$. 
Note that, unless $\nohyp$, a {\sf para}-{\nph} problem cannot be in \xp, hence it cannot be {\fpt} either.

Within parameterized problems, the class {\wone} may be seen as the parameterized equivalent to the class {\np} of classical optimization problems. 
Without entering into details (see~\cite{CyganFKLMPPS15,DF13} for the formal definitions), a parameterized problem being {\wonehard} can be seen as a strong evidence that this problem is {\sl not} \fpt. 
To transfer {\sf W}[1]-hardness from one problem to another, one uses an $\fpt$ reduction, which given an input $I=(x,k)$ of the source problem, computes in time $f(k) \cdot |I|^c$, for some computable function $f$ and a function $c$, an equivalent instance $I'=(x',k')$ of the target problem such that $k'$ is bounded by a function depending only on $k$. 
The following problem is a {\wonehard} problem that we will use in our reductions.

\paraprobl
{\mcc}
{A graph $G$, an integer $k$, a coloring function $\chi: V(G) \to [k]$.}
{Does $G$ contain a clique on $k$ vertices with one vertex from each color class?}
{$k$.}

{\mcc} is known to be {\wonehard} on general graphs, even in the special case where all color classes have the same number of
vertices~\cite{Pietrzak03}. 
Clearly, we can also assume that every color class is an independent set since the problem is indifferent to edges within the same color class. 

\vspace{2mm}\runningtitle{Tree decompositions} 
A \emph{tree decomposition} of a graph $G=(V,E)$ is a pair $\DD=(T,\XX)$, where $T$ is a tree
and $\XX= \{ X_t \mid t \in V(T) \}$ is a collection of subsets of $V(G)$ such that:

\begin{itemize}
\item $\bigcup_{t \in V(T)} X_t = V$,
\item for every edge $u v \in E$, there is a $t \in V(T)$ such that $\set{u, v} \subseteq X_t$, and
\item for every $\set{ x, y, z} \subseteq V(T)$ such that $z$ lies on the unique path between $x$ and $y$ in $T$,  $X_x \cap X_y \subseteq X_z$.
\end{itemize}
We call the vertices of $T$ {\em nodes} of $\DD$ and the sets in $\XX$ {\em bags} of $\DD$. 
The width of $\DD$ is $\max_{t \in V(T)} |X_t| - 1$.
The \emph{treewidth} of $G$, denoted by $\tw(G)$, is the smallest integer $w$ such that there exists a tree decomposition of $G$ of width $w$.
A tree decomposition in which the tree $T$ is restricted to be a path is called a \emph{path decomposition}. 
The \emph{pathwidth} of $G$, denoted by $\pw(G)$, is the smallest integer $w$ such that there exists a path decomposition of $G$ of width $w$.

A tree decomposition is {\em rooted} if we distinguish in $T$ some specific vertex $r$, and consider $T$ as a rooted (on $r$) tree.
We denote such a tree decomposition by a triple $\DD=(T,\XX,r)$.

\vspace{2mm}\runningtitle{Nice tree decompositions} 
Let $\DD=(T,\XX, r)$ be a rooted tree decomposition of $G$, 
and  $\GG = \{ G_t \mid t \in V(T) \}$ be a collection of subgraphs of $G$.
We say that the ordered pair $(\DD,\GG)$ is \emph{nice} if the following conditions hold:

\begin{itemize}
\item $X_r = \es$ and $G_r=G$,
\item every node of $\DD$ has at most two children in $T$,
\item for each leaf $t \in V(T)$, $X_t = \es$ and $G_t=(\es,\es).$ Such a node $t$ is called a {\em leaf node},
\item if $t \in V(T)$ has exactly two children $t'$ and $t''$, then $X_t = X_{t'} = X_{t''}$, $G_t = G_{t'} \cup G_{t''}$, and $E(G_{t'}) \cap E(G_{t''}) = \emptyset$. 
The node $t$ is called a \emph{join node}.
\item if $t \in V(T)$ has exactly one child $t'$, then exactly one of the following holds.

\begin{itemize}
\item $X_t = X_{t'} + v$ for some $v \notin X_{t'}$ and $G_t=(V(G_{t'}) + v , E(G_{t'}))$.
The node $t$ is called \emph{vertex-introduce node} and the vertex $v$ is the {\em introduced vertex} of $X_t$.

\item $X_t = X_{t'}$ and $G_t=(G_{t'},E(G_{t'}) + e)$ where $e$ is an edge of $G$ with endpoints in $X_t$.
The node $t$ is called \emph{edge-introduce node} and the edge $e$ is the {\em introduced edge} of $X_t$.

\item $X_t = X_{t'} - v$ for some $v \in X_{t'}$ and $G_t=G_{t'}$.
The node $t$ is called \emph{forget node} and $v$ is the {\em forget vertex} of $X_t$.
\end{itemize}
\end{itemize}

The notion of a nice pair defined above is essentially the same
as the one of nice tree decomposition in~\cite{CyganNPPRW11} (which in turn is an enhancement of the original one, introduced in~\cite{Klo94}).
As already argued in~\cite{CyganNPPRW11,Klo94}, given a tree decomposition, it is possible to transform it in polynomial time
to a tree decomposition $\DD$ of the same width and construct a collection $\GG$ such that $(\DD, \GG)$ is nice.

\vspace{2mm}\runningtitle{Reload cost model} 
For reload costs, we follow the notation and terminology defined by \cite{WiSt01}. 
We consider an edge-colored graph $G$ where edge colors are taken from a finite set $X$ and 
the coloring function is $\chi: E(G) \rightarrow X$. The reload costs are given by a function $c:X^2 \rightarrow  \mathbb{N}_{\geq 0}$ where 
$c(x_{1},x_{2})=c(x_{2},x_{1})$ for each $(x_{1},x_{2})\in X^{2}$. 
The cost of {\em traversing} two incident edges $e_1$, $e_2$ of $G$ is $\tc(e_1, e_2)=c(\chi(e_1), \chi(e_2))$. 
Given a subgraph $H$ of $G$ and a vertex $v\in V(H)$, we define the {\em reload cost}
of $v$ in $H$ as 
$\rc_{\chi,c}(H,v)=\sum_{\{e_1,e_2\}\in{E_{H}(v)\choose 2}} {\tc}(e_1,e_2)$
and the {\em reload cost} of $H$ as 
$\rc_{\chi,c}(H)=\sum_{v\in V(H)}\rc_{\chi,c}(H,v)$.
When $\chi$ and $c$ are clear from the context, we write $\rc(v)$ and $\rc(H)$ instead.

\vspace{2mm}\runningtitle{Problem statement} 
The problem we study in this paper can be formally defined as follows for every $r \in \mathbb{N}_{\geq 2}$:

\probl
{Minimum Reload Cost $r$-Factor (\rmrcf)}
{A graph $G$, an edge-coloring $\chi$, a reload cost function $c$, and a non-negative integer $k$.}
{Is there an $r$-factor $H$ of $G$ with reload cost at most $k$, i.e., $\rc(H) \leq k$?}

Given an instance $\inst$ of $r$-MRCF we consider the following parameters:
\begin{itemize}
\item the maximum degree $\Delta(G)$ of $G$, 
\item the treewidth $\tw(G)$ of $G$,
\item the number of colors $q=|X|$,
\item the number of distinct costs: $d=|\{c(x_1,x_2)\mid (x_{1},x_{2})\in X^2\}|$,
\item the minimum traversal cost $c_{min}=\min\{c(x_1,x_2)\mid (x_{1},x_{2})\in X^2\}$, and
\item the total cost $k$.
\end{itemize}
Clearly, we can assume that $\Delta(G) \geq r+1$, and also $\delta(G) \geq r$ since otherwise the instance is trivial.
Let $\kmin = c_{min} \cdot |V(G)| \cdot {r \choose 2}$.
Note that the reload cost of every $r$-factor is at least $\kmin$.
Therefore, we can also assume that $k \geq k_{\rm min}$.

A summary of the results regarding the classical and parameterized complexity of the $r$-MRCF problem is shown in Table \ref{writings} and Table \ref{magnetic}, respectively.

\begin{table}[!htbp]
\begin{center}
\begin{tabular}{ |c|c|c|c|c|l| } 
\hline
$G$       & $\Delta(G)$  & $d$   & $k$       & $q$           &  \\
\hline
          & $\le r+2$ & 2     & $k_{min}$ & $\min\{r,3\}$ & $\nph$ (Theorem \ref{consiste})\\  
\hline
          & $r+1$     &       &           &               & Polynomial (Theorem \ref{consists})\\ 
\hline
 Planar   & $\le r+4$ & 2     & $k_{min}$ & 7             & $\nph$ (Theorem \ref{adorning}) \\  
\hline
\end{tabular}
\caption{Summary of classical complexity results for the $r$-MRCF problem.}\label{writings}
\end{center}

\begin{center}
\begin{tabular}{|c|c|c|c|c|l|} 
\hline
Parameter & $d$   &   $k$     & average degree  & \\ 
\hline
$\pw(G)$  &   2   & $k_{min}$ & $< r + \left\{ \begin{array}{ll}
                                4 & \textrm{~if~} r=2 \\ 4/3 & \textrm{otherwise}. 
                                \end{array}\right.$ & {\wonehard} (Theorem \ref{dispense})\\ 
\hline
$ \tw(G) + \min \{q, \Delta(G)\}$ &  & & & \fpt (Theorem \ref{softened})\\ 
\hline
\end{tabular}
\caption{Summary of parameterized complexity results for the $r$-MRCF problem.}\label{magnetic}
\end{center}
\end{table}

\section{Classical Complexity of $r$-MRCF}
The following construct will be used in our reductions.
A {\em diamond} is a graph on four vertices and five edges, that is obtained by adding one chord to a cycle on four vertices.
Clearly, a diamond contains two vertices of degree two and two vertices of degree three.
The degree two vertices are termed as the {\em tip}s of the diamond.
A {\em joker} is a monochromatic diamond, that is, a diamond whose edges have the same color.
In our reductions, every a joker $J$ will have exactly one vertex adjacent to other vertices of the graph.
This vertex will always be a tip of $J$, and we will term it as the \emph{connecting tip} of $J$.
Given a joker $J$ and a 2-factor $F$, it is easy to see that exactly one of the following happens:
\begin{itemize}
\item{} $F \cap E(J)$ is the $4$-cycle of $J$, and $F \sm E(J)$ does not contain any edges incident to $J$.
\item{} $F \cap E(J)$ is a triangle of $J$, $F \sm E(J)$ contains exactly two edges incident to $J$ both of which are incident its connecting tip.
\end{itemize}
Furthemore, since our cost functions satisfy $c(\lambda,\lambda)=0$ for every color $\lambda$, and $F \cap E(J)$ is always a cycle, the joker edges do not affect the cost of $F$. When we describe a $2$-factor $F$, these properties allow us to leave the edges $F \cap E(J)$ unspecified since they are implied by the edges of $F \sm E(J)$.
Such a partial description is valid if and only if the connecting tip have degree zero or two. 
Finally, a 5-\emph{joker} is a cycle on five vertices with an added chord.
This graph has one triangle with one degree 2 vertex that we will refer to as \emph{the tip} of the $5$-joker.
Note that a $5$-joker has all the properties of a joker. 

Another construct that we use in our reductions is a graph $Q_\ell$ that is obtained from the clique on $\ell+1$ vertices by subdividing $\ell-2$ arbitrary edges twice (into three edges) and removing the middle edge of each one.
Clearly, $Q_\ell$ contains $\ell+1$ vertices of degree $\ell$ and $2(\ell-2)$ vertices of degree one. In total $G_{\ell}$ has $3\ell-3$ vertices.

Galbiati et.al. proved in \cite{GGM14} that {\mrcf{2}}, a.k.a. the minimum reload cost cycle cover problem, is $\nph$ even when the number of colors is $2$, the reload costs are symmetric and satisfy the triangle inequality. In the following, we obtain a hardness result for {\rmrcf}, which in particular implies a stronger hardness result for the special case of $r=2$. 

\begin{theorem}\label{consiste}
{\rmrcf} is $\nph$ for every $r \geq 2$ even when $\Delta(G) \leq r+2$, $d=2$, $k=\kmin=0$ and $q=\min \{r, 3\}$.
\end{theorem}

\begin{proof}
We start by proving the claim for $r=2$, that is, we prove that \mrcf{2} is $\nph$ when $\Delta(G) \leq 4$, $d=q=2$, $k=\kmin$.
The proof is by reduction from the \OneThreeSAT problem which is known to be $\nph$ \cite{Schaefer78}.
An instance of \OneThreeSAT is an expression $\phi$ that consists of $n$ boolean variables $x_1, \ldots, x_n$ and $m$  clauses $c_1, \ldots, c_m$. 
Every clause $c_j$ consists of three literals $\ell_{j,1}, \ell_{j,2}, \ell_{j,3}$, each of which is either an occurrence of some variable $x_i$ or its negation $\bar{x_i}$. 
A clause is satisfied by a truth assignment if exactly one of its literals is satisfied by it, and $\phi$ is satisfied if all its clauses are satisfied.
\OneThreeSAT is the problem of determining whether a given expression $\phi$ has a satisfying assignment.

Given an instance $\phi$ as described above, we construct an instance $\inst$ of \mrcf{2}.
The set of colors is $X=\{red, blue\}$, thus $q=2$.
The traversal cost between two identical colors is zero and the traversal cost between red and blue is $1$, thus $d=2$.
Furthermore, $k=0$. 
We have $c_{min}=0$, which implies that $\kmin=0=k$.
To complete the reduction, we construct a graph $G$ with $\Delta(G)=4$ whose edges are colored with colors from $X$.

The graph $G$ is the union of $n$ vertex disjoint variable gadgets with vertex sets $X_1,\ldots,X_n$,
$m$ vertex disjoint clause gadgets with vertex sets $C_1, \ldots, C_m$, 
and $3m+4n$ vertex disjoint jokers, that is, four jokers $J_i^{(u)}$, $J_i^{(v)}$, $J_i^{(2)}$, $J_i^{(4)}$ for every variable $x_i$ and a joker $J_{i,k'}$ for every occurrence of it.
In our construction, every edge will join two vertices of the same gadget.
However, we will allow a vertex to appear in two (but not more) of the above gadgets.
Every gadget will have maximum degree 4, where vertices shared with another gadget have degree only 2 inside the gadget.
This will ensure that $\Delta(G)=4$. We now proceed with the description of the individual gadgets.

Let $x_i$ be a variable that has $p$ occurrences in $\phi$. 
Consult Figure \ref{narcotic} for the description of the corresponding variable gadget. 
We set $X_i = U_i \cup V_i \cup W_i \cup Z_i^+ \cup Z_i^-$ where
$U_i=\{u_{i,0}, u_{i,1}, \ldots, u_{i,p}\}$, $V_i=\{v_{i,0}, v_{i,1}, \ldots, v_{i,p}\}$, $W_i = \{w_{i,1}, w_{i,2}, w_{i,3}, w_{i,4} \}$
$Z_i^D=\{z_{i,1}^D, \ldots, z_{i,p}^D\}$ for $D \in \{+,-\}$. 
The edge set induced by $X_i$ is the disjoint union of $p$ red triangles, $p$ blue triangles, 
a blue cycle on the four vertices $u_{i,0}, w_{i,1}, w_{i,2}, w_{i,3}$ and a red cycle on the four vertices $v_{i,0}, w_{i,1}, w_{i,4}, w_{i,3}$ 
as depicted in Figure \ref{narcotic}. 
Note that every vertex of $Z_i^+$ (resp. $Z_i^-$) has degree two, and both its incident edges are red (resp. blue).

For every clause $c_j$, the gadget $C_i$ has six vertices $a_j, b_j, d_j, s_{j,1}, s_{j,2}, s_{j,3}$ connected with blue edges as depicted in Figure \ref{unguided}.

We identify each of the vertices $w_{i,2}$, $w_{i,4}$, $u_{i,k}$ and $v_{i,k}$ with a tips of the jokers $J_i^{(2)}$, $J_i^{(4)}$, $J_i^{(u)}$, and $J_i^{(v)}$, respectively.
Finally, for every literal $\ell_{j,k}$ that is the $k'$-th occurrence of some variable $x_i$ we identify the vertex $s_{j,k}$ with the vertex $z_{i,k'}^+$ (resp. $z_{i,k'}^-$) if the occurrence is positive (resp. negative) and we identify the vertex $z_{i,k'}^-$ (resp. $z_{i,k'}^+$) with a tip of the joker that corresponds to this literal.

\begin{figure}
\begin{center}
\commentfig{\includegraphics[width=15.8cm]{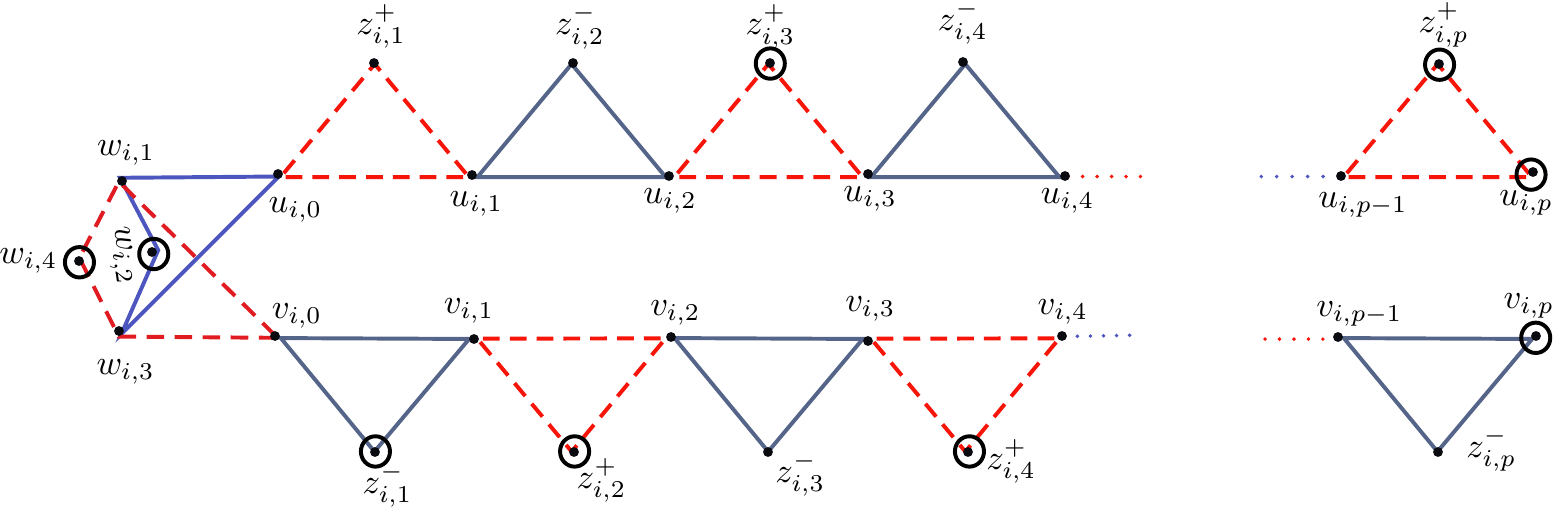}}
\caption{The variable gadget in Theorem \ref{consiste} corresponding to a variable $x_i$ having $p$ occurrences in $\phi$. 
The solid edges are blue, the dashed edges are red. 
Every vertex surrounded with a circle is the connecting tip of a joker. 
For every $k' \in [p]$ exactly one of $z_{i,k'}^+, z_{i,k'}^-$ is a connecting tip.
The gadget contains $p+4$ connecting tips.} \label{narcotic}
\end{center}
\end{figure}

To complete this part of the proof, it remains to show that $\phi$ has a satisfying assignment if and only if $G$ contains a $2$-factor $H$ with zero reload cost.
We suppose that $\phi$ has a satisfying assignment and provide a partial description of the 2-factor $H$ incrementally.
Recall that a partial description is valid if and only if every vertex has degree two, except possibly some connecting tips that are isolated.
Let $x_i$ be a variable assigned true by the assignment.
$H[X_i]$ consists of all the blue edges of $G[X_i]$.
Clearly, the traversal cost between these edges is zero.
It is also easy to see that all the vertices of $H[X_i]$ have degree two except the vertices $Z_i^+$, $w_{i,4}$ and one of $u_{i,p}, v_{i,p}$ (depending on the parity of $p$) all of which are isolated.
Note that $u_{i,p}, v_{i,p}$, $w_{i,4}$ and all the vertices of $Z_i^+$ corresponding to negative occurrences of $x_i$ are connecting tips of some joker.
We remain with the vertices of $Z_i^+$ corresponding to positive occurrences of $x_i$ that are isolated but should have degree two in order to make (the partial description of) $H$ valid.
Proceeding in the same way for the variables assigned false, we get a subgraph $H$ in which all the vertices of the variable gadgets which are not connecting tips have degree two, 
except the vertices of $Z_i^+$ corresponding to positive occurrences of variables assigned true and the vertices of $Z_i^-$ corresponding to negative occurrences of variables assigned false.
In other words, these are exactly the vertices of $Z_i^+ \cup Z_i^-$ that correspond to literals satisfied by the assignment.
By adding to $H$ the cycle $a_j d_j b_j s_{j,p}$ for every clause $c_j$ (where $\ell_{j,p}$ is its unique literal satisfied by the assignment),
we get a $2$-factor $H$ of zero cost.

\begin{figure}
\begin{center}
\commentfig{\includegraphics[scale=0.690]{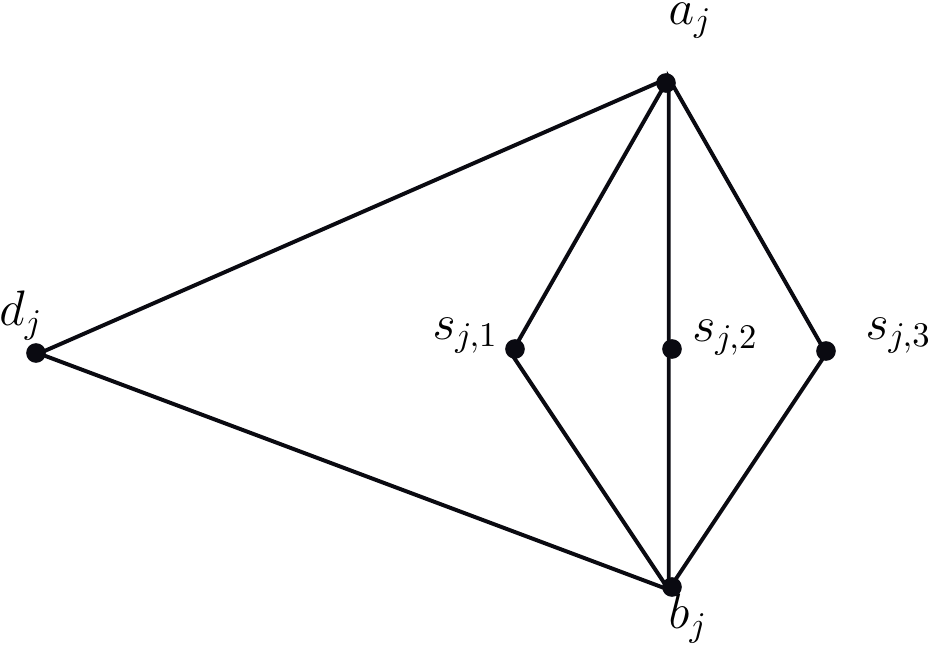}}
\caption{The clause gadget corresponding to the clause $c_j$.}\label{unguided}
\end{center}
\end{figure}

Conversely, suppose that $G$ contains a $2$-factor $H$ of zero reload cost.
Consider the vertices $U_i$ of the gadget $X_i$. 
Since every such vertex has two incident blue edges and two incident red edges, $H$ contains either the two blue edges or the two red edges.
We conclude that the vertices of $U_i - u_{i,0}$ are covered either by all the red triangles incident to them or all the blue ones.
By symmetry, the same holds for the vertices of $V_i - v_{i,0}$.
However, we note that the choice of the color (blue or red) must be identical in both $U_i$ and $V_i$ since otherwise a positive traversal cost is incurred at each of the vertices $w_{i,1}, w_{i,3}$.
Therefore, $H[X_i]$ is the subgraph of $G[X_i]$ that consists of all the blue (or all the red) edges.
We assign true to $x_i$ if $H[X_i]$ contains the blue edges and false otherwise.
It remains to show that this assignment satisfies $\phi$.
Consider a clause gadget $C_j$.
$H$ must contain the two edges $d_j a_j$ and $d_j b_j$.
Furthermore, since every vertex $s_{j,k}$ is identified with a $z$ vertex, $H$ contains either both edges of the clause gadget incident to $s_{j,k}$ or none of them.
Since $a_j$ has degree 2 in $H$, exactly one of these three vertices, say $s_{j,1}$, has both incident edges of the clause gadget in $H$.
Then, the two edges incident to $s_{j,1}$ in the corresponding variable gadget are not in $H$.
This implies that either $\ell_{j, 1}$ is a positive occurrence of a positive variable, or a negative occurrence of a negative variable.
In other words, $\ell_{j, 1}$ is satisfied.
The vertices $s_{j,2}$ and $s_{j,3}$ must have in $H$ their incident (red or blue) edges in their corresponding variable gadgets. 
Then, each one of them is either a positive occurrence of a negative variable, or a negative occurrence of a positive variable.
In other words, their corresponding literals are not satisfied.
We conclude that $C_j$ is satisfied, and thus $\phi$ is satisfied.

We now prove the claim for any $r \geq 3$. 
We need to show that \rmrcf~is $\nph$ even when $t=2$, $q=3$, $k=\kmin=0$ and $\Delta \leq r+2$.
Given an instance $\inst$ of \mrcf{2} with $\Delta \leq 4$, $t=q=2$ and $k=\kmin=0$,
we construct an instance $\instprime$ of \rmrcf~as follows.
By eventually adding a monochromatic cycle of size $3$, we may assume that $G$ has an even number of vertices.
The color set $X'$ is $X \cup \{ white\}$ where white is not a color of $X$. Thus, $q=3$. White edges are precisely the edges $E(G')\setminus E(G)$. The traversal cost $c'$ between white and every color is zero, and for any two colors of $X$, the traversal cost $c'$ is the same as $c$. 
Thus, $t=2$ and $c_{min}=0$, implying that $\kmin=0=k$.
Let $\mathcal{P}$ be a partition of the vertices $V(G)$ into pairs.
$G'$ is obtained from $G$ by adding to $G$ a $Q_r$ for every pair $\{u,v\} \in \mathcal{P}$ and identifying $r-2$ of its degree one vertices and $u$ altogether with each other as well as identifying the remaining $r-2$ vertices and $v$ with each other. Note that the pair $\{u,v\}$ is not necessarily an edge of the graph.
Every vertex of $V(G') \cap V(G)$ has $r-2$ white incident edges and at most 4 other incident edges since $\Delta(G) \le 4$, for a total of at most $r+2$ edges.
The degree of every vertex of $V(G') \sm V(G)$ is $r$. 
Therefore, $\Delta(G') \leq r+2$.
Let $H'$ be an $r$-factor of $G'$.
As the vertices of $V(G') \sm V(G)$ are of degree exactly $r$ in $G'$, $H'$ contains all the edges incident to the vertices of $V(G') \sm V(G)$, i.e. all the white edges.
By removing these edges from $H'$ we get a $2$-factor $H$ of $G$.
Moreover, $\rc(H)=\rc(H')$ since the white edges do not introduce a traversal cost.
We conclude that $\inst$ has a $2$-factor of zero cost if and only if $\instprime$ has an $r$-factor of zero cost. 
\end{proof}

The value of the parameter $t$ in Theorem \ref{consiste} is clearly tight.
When $t=1$, that is, when there is only one traversal cost, all the $r$-factors have the same reload cost.
In this case the problem reduces to determining the existence of an $r$-factor.
This problem is known to be polynomial-time solvable \cite{Pul73}.
The following theorem states that the parameter $\Delta(G)$ of Theorem \ref{consiste} is also tight.

\begin{theorem}\label{consists}
For every $r \in \mathbb{N}_{\geq 2}$, if $\Delta(G) = r+1$, then {\rmrcf} can be solved in polynomial time.
\end{theorem}

\begin{proof}
Recall that we assume $\delta(G) \geq r$. 
Let $R$ be the vertices of degree $r$ of $G$ and $R^+ = V(G) \sm R$ the vertices of degree $r+1$.
Let also $H$ be an $r$-factor of $G$.
We observe that $H$ is obtained by removing a perfect matching $M$ of $G[R^+]$ from $G$.
Every edge $e \in M$ reduces the reload cost by the sum of the traversal costs with all its incident edges.
Therefore, $\rc(H) = \rc(G) - \sum_{e \in M} w(e)$ where

\[
w(e) \defined \sum_{e' \in E_G(u) - e} c(e,e') + \sum_{e' \in E_G(v) - e} c(e,e')
\]

\noindent for every edge $e=uv$ of $G$.
Then, minimizing $\rc(H)$ boils down to the problem of finding a maximum weight perfect matching $M$ of $G[R^+]$ with the edge weight function $w$, which can clearly be solved in polynomial time.
\end{proof}

Both the result in \cite{GGM14} and Theorem \ref{consiste} are for general graphs and hence leave the complexity of the problem open for special graph classes. In the following (in Theorem \ref{adorning}), we prove hardness of {\rmrcf} in planar graphs even under restricted cases.
We need the following Lemma.

\begin{lemma}\label{appetite}
Let $G$ be a planar graph with an even number of vertices and $\delta(G) \geq 2$.
There is a partition  $M$ of $V(G)$ into pairs such that the multigraph $G'$ obtained by adding to $G$ an edge between every pair of $M$ is planar.
Moreover, such a partition can be found in polynomial time.
\end{lemma}

%
%
\begin{proof}
We construct $M$ incrementally, starting from $M = \es$.
At any given point of the execution of the algorithm $M$ is a set of vertex-disjoint pairs of vertices.
We add pairs to $M$ until $\cup M = V(G)$, implying that $M$ is a partition of $V(G)$.   
A vertex $v$ is \emph{paired} (resp. \emph{unpaired}) if $v \in \cup M$ (resp. $v \notin \cup M$). 
We use a sweep line algorithm \footnote{A well known technique in computational geometry \cite{Berg2008computationalgeometry}.} that scans the embedding from left to right. 
Whenever the sweep line encounters a new inner face $F$, it pairs all the unpaired vertices of $F$, except possibly one vertex in case the number of unpaired vertices is odd. The pairing is done such that the pairs are non-crossing,
i.e. such that one can add edges to $G$ between the pairs while preserving planarity.
Such a pairing can be easily obtained as follows.
Let $v_1, v_2, \ldots, v_{2 \ell}$ be the vertices to be paired, in the clockwise order they appear on the face $F$.
The pairing, $\{ \set{v_1, v_2}, \ldots, \set{v_{2 \ell - 1}, v_{2 \ell}} \}$ is a non-crossing pairing.
If the number of unpaired vertices of $F$ is odd, the rightmost unpaired vertex of $F$ is left unpaired.

When the sweeping phase is terminated, all the unpaired vertices are on the outer face, and their number is even.
The algorithm terminates after processing the outer face exactly in the same way as an inner face was processed.
\end{proof}

\begin{theorem} \label{adorning}
For every  $r \in [2,5]$, {\rmrcf} is {\nph} even when the input graph $G$ is planar, 
$\Delta(G) \leq r+4$, $d=2$, $k=\kmin$, and $q=\left \{ \begin{array}{ll} 6 & \textrm{if~} r = 2\\ 7 & \textrm{otherwise}. \end{array} \right.$
\end{theorem}

\begin{proof}
We start by proving the theorem for $r=2$
by reducing an instance $\inst$ of {\mrcf{2}} with $\Delta(G) \leq 4$, $d=2$, $k=\kmin$, $q=2$  to 
an instance $\instprime$ of {\mrcf{2}}, with $\Delta(G') \leq 6$, $d=2$, $k=\kmin$, $q=6$, and $G'$ is planar.

We rename the colors of the original graph as $1$ and $2$, and we add four colors $X_C= \{red,$$ blue,$ $ green, yellow \}$.
Therefore, $q=6$. 
As for the cost function we use the following function $c'$ that is an extension of $c$ and uses only costs from $\{0, 1\}$, i.e. $d=2$.
\[
c'(\lambda,\lambda') = \left\{ \begin{array}{ll}
0 & \textrm{if}~\lambda=\lambda'\\
0 & \textrm{if}~yellow \in \{\lambda,\lambda'\}~\textrm{and}~\{1,2\} \cap \{ \lambda,\lambda' \} \neq \emptyset\\
0 & \textrm{if}~1 \in \{\lambda,\lambda'\}~\textrm{and}~2 \notin \{ \lambda,\lambda' \}\\
1 & \textrm{otherwise.}
\end{array} \right.
\]
We consider a planar embedding of $G$ where all crossings are polynomial on $|V(G)|$ and no three edges cross the same point. 
$G'$ is obtained by replacing every crossing point of two edges $e,e'$ of $G$ by a copy of the gadget depicted in Figure~\ref{skeleton}, 
and 9 jokers such that one tip of each joker is identified with a distinct vertex of the gadget.
We observe that $\Delta(G') \leq 6$, as required.
For a gadget under consideration,
the set of vertices that contain the index $L$ (resp. $R$) are its \emph{left} (resp. \emph{right}) vertices,
and the remaining vertices are its \emph{middle} vertices.
The \emph{left part} (resp. \emph{right part}) of the gadget is the subgraph induced by the left (resp. right) and middle vertices together.
We refer to a four cycle $m_P a_{P,T} t_P a_{P,B}$ for $P \in \set{L,R}$ as a \emph{green cycle},
to the triangle $t_T t_L t_R$ as the \emph{middle triangle},
to a triangle that consists of an $m$-vertex, a $t$-vertex and an $a$-vertex as a \emph{blue} or \emph{red} triangle,
and to the 8-cycle induced by the four $c$-vertices together with the four $m$-vertices as the \emph{yellow rectangle}.
Note that all these cycles have zero reload cost, though (despite their names) they are not monochromatic.

The two parts of the edge $e$ (resp. $e'$) inherit their color from $e$ (resp. $e'$).
Clearly, $V(G) \subseteq V(G')$.
Finally, $k'=0$.

\begin{figure}[h]
\begin{center}
\scalebox{.96}{\commentfig{\includegraphics{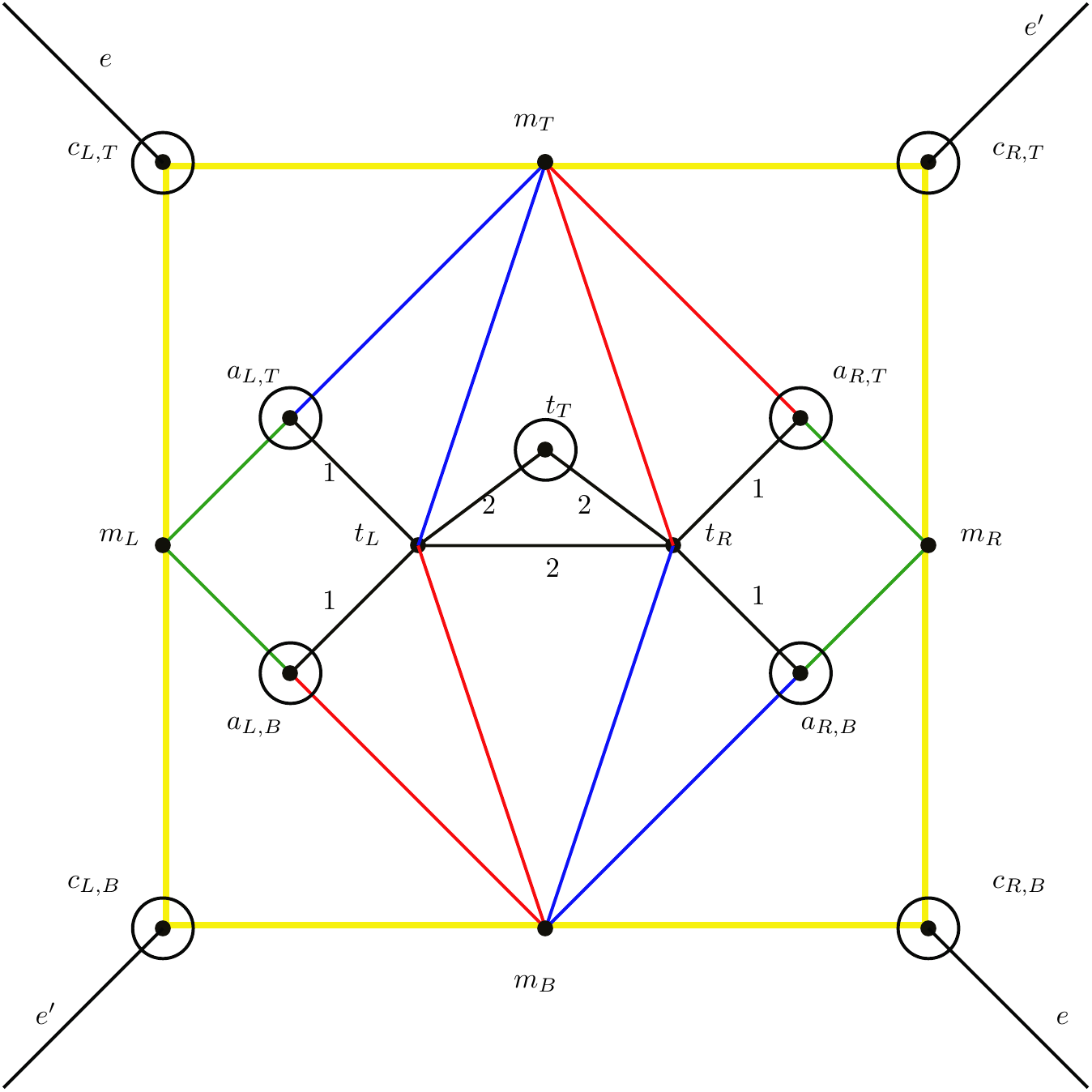}}}
\end{center}
\caption{The planar gadget that replaces two crossing edges of $G$. The vertices marked with a circle are identified with a tip of a joker.} \label{skeleton}
\end{figure}

It remains to show that $G$ has a $2$-factor with zero reload cost if and only if $G'$ has one.
Suppose that $G$ has a $2$-factor $F$ with zero cost.
We construct a $2$-factor $F'$ of $G'$ as follows.
 
A part of an original edge is in $F'$ if and only if the original edge is in $F$.
Since $F$ is a $2-$factor of zero cost, all the original vertices have degree 2 in $F'$ and they incur a zero reload cost. 
Recall that we allow for the vertex of a gadget that is the connecting tip of a joker to be isolated in our description of $F'$ since such a vertex will be in a cycle of the joker in $F'$.
For every gadget that corresponds to a crossing point of two edges $e$ and $e'$, we proceed as follows
to ensure that $F'$ incurs a zero cost in every gadget.
\begin{itemize}
\item {$\{e,e'\} \cap F = \emptyset$:} 
In this case $F'$ contains the yellow rectangle and the middle triangle. Then $F'$ incurs a zero cost in this gadget.

\item {$\{e,e'\} \subseteq F$:} 
In this case $F'$ contains the two green cycles
and the two paths $c_{L,H} m_H c_{R,H}$ for $H \in \{ T, B \}$.
Combining with the two parts of $e$ and $e'$ incident to this gadget, the degree of every vertex of this gadget in $F'$ is $2$.
We recall that $c'(1,yellow)=c'(2,yellow)=0$ and conclude that $F'$ incurs a zero cost in this gadget.

\item {$|\{e,e'\} \cap F| = 1$:} 
Assume without loss of generality that $e \in F$ and $e' \notin F$. 
In this case $F'$ contains the left blue triangle, the right green cycle,
and the two adjacent edges of the yellow rectangle, 
namely the path $c_{L,T} m_L c_{L,B} m_B c_{R,B}$. 
Combining with the two parts of $e$ incident to this gadget, the degree of every vertex of this gadget in $F'$ is $2$.
We note that $c'(\chi(e),yellow)=0$ and conclude that $F'$ incurs a zero cost in this gadget.
\end{itemize}

Conversely, suppose that $G'$ has a $2$-factor $F'$ with zero cost.
Recall that for any two distinct colors $\lambda,\lambda' \in \{red, blue, green, yellow\}$ we have $c'(\lambda,\lambda'')=1$.
Therefore, for every $m$-vertex the $2$-factor $F'$ contains two edges of the same color.
If $F'$ contains two adjacent blue (resp. red) edges, it must contain the third edge colored $1$ that completes the blue (resp. red) triangle because otherwise an $a$-vertex incurs a positive cost, since $c'(green,red)=c'(green,blue)=1$.
If $F'$ contains two adjacent green edges, then for the same reason, $F'$ contains the corresponding green cycle.

Consider one of the parts, say the left part of a gadget. 
Since $t_L$ is common to the blue, red and green cycles of this part,
at most one of these cycles are in $F'$.
We consider two cases:
\begin{itemize}
\item {$F'$ contains none of the cycles on the left part:} 
Then the edges $t_L t_R$ and $t_L t_T$ are in $F'$,
implying that none of the cycles on the right part are in $F'$.
Since all the $m$ vertices have degree 2 in $F'$ and none is covered by the inner cycles, $F'$ must contain the yellow rectangle.
Then none of the parts of $e$ and $e'$ are in $F'$.

\item {$F'$ contains one cycle on the left part:}
In this case, since the degree of $t_{R}$ must be $2$ in $F'$, $F'$ contains one cycle of the right part. Every cycle saturates one of the $m$ vertices, for a total of two $m$ vertices saturated.
Then, the remaining two $m$ vertices are saturated by four yellow edges that together form two edges of the yellow rectangle.
We consider two sub-cases
\begin{itemize}
\item{$F'$ contains two non-adjacent edges of the yellow rectangle.}
In this case all the four yellow edges have degree 1.
Therefore, both parts of $e$ and both parts of $e'$ are in $F'$.

\item{$F'$ contains two adjacent edges of the yellow rectangle.}
In this case two opposite corners of the yellow cycle, say $c_{L,T}$ and $c_{R,B}$ have one incident yellow edge, while the other two corners have either two or zero incident yellow edges.
Then both parts of $e$ (in the symmetric case, of $e'$) are in $F'$,
and none of the parts of the other edge is in $F'$.
\end{itemize}
\end{itemize}
We conclude that for every edge $e$ of $G$, either all its parts or none are in $F'$. 
Then $F'$ implies a subgraph $F$ of $G$.
Since all the original vertices have degree 2 in $F'$, they have degree 2 in $F$, i.e. $F$ is a 2-factor of $G$.
Furthermore, the cost at every original vertex is zero since $c'$ extends $c$.

Now for any $r \in [3,5]$, we have to reduce an instance $\inst$ of \mrcf{2} where $\Delta(G) \leq 6$, $d=2$, $k=\kmin$, $q=6$, to an instance $\instprime$ of {\rmrcf}, where $\Delta(G') \leq r+4$, $d=2$, $k=\kmin$, $q=7$ and $G'$ is planar. 
Note that, since whenever $G$ is planar, we have $\delta(G) \leq 5$, it does not make sense to consider other values of $r$.

If $G$ is odd, we subdivide an edge and add a $5$-joker whose tip is identified with the newly created vertex.
We color all the edges with the color of the original edge.
This creates an equivalent instance having all the properties assumed for $G$.
By Lemma \ref{appetite}, there is a pairing $M$ of the vertices of $G$ such that edges can be added between each pair of vertices while still preserving planarity.
For every pair $\set{u,v} \in M$ we will add a gadget that will have the desired properties of $Q_r$, and in addition will preserve planarity.
Namely, the gadget will have $2(r-2)$ vertices of degree one, all the remaining vertices of degree $r$, and will have a planar embedding in which all the degree one vertices are in the outer face.
Half (i.e., $(r-2)$) of the degree one vertices of the gadget are identified with $u$ and the remaining half are identified with $v$.
Now every new vertex has degree $r$ and the degree of every original vertex increased by $r-2$.
The new edges are colored with a new color $\lambda'$ such that for every color $\lambda$ we have $c'(\lambda',\lambda)=0$.
$\instprime$ has an $r$-factor $F'$ such that $\rc(F')=0$ if and only if $\inst$ has a $2$-factor $F$ with $\rc(F)=0$.
We omit the proof of this fact since it is identical to the non-planar case. 

We now describe the gadget depending on the value of $r$.
\begin{itemize}
\item {$r=3$.} 
In this case the gadget is a $Q_3$ which has the desired planarity properties.

\item {$r=4$.} 
In this case the gadget is a $Q_4$. 
However, in order to get the planarity properties, the subdivided edges of the $K_5$ should constitute a matching.

\item {$r=5$.}
In this case we use the gadget depicted in Figure \ref{friesian}.
\end{itemize}

\begin{figure}
\begin{center}
\scalebox{.6}{\commentfig{\includegraphics{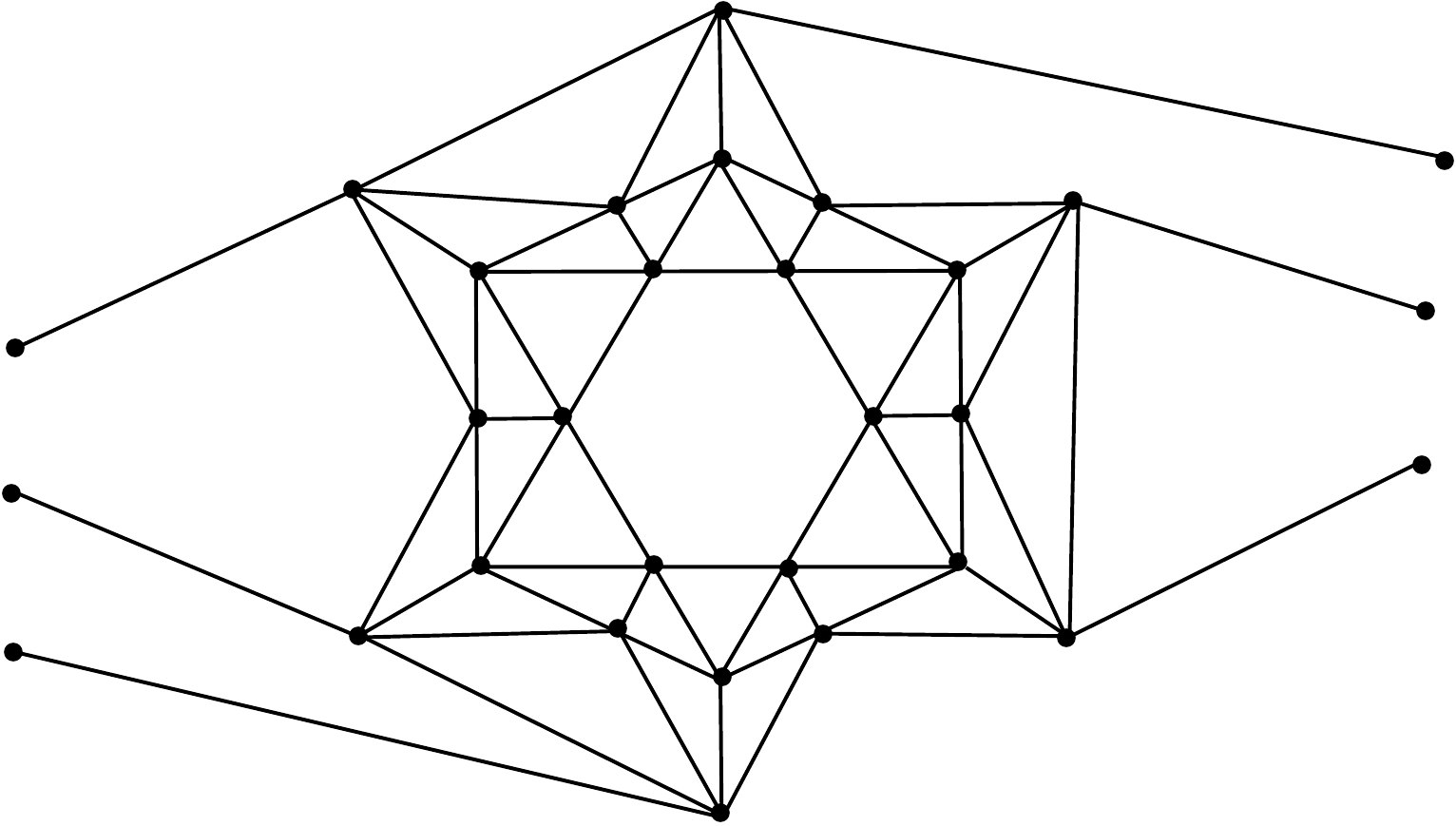}}}
\end{center}
\caption{The planar gadget used for the case $r=5$.}\label{friesian}
\end{figure} 

\end{proof}


\section{Parameterized Complexity of $r$-MRCF}
\begin{lemma}\label{granting}
For every $\alpha>1$,  \mrcf{2}, parameterized by $\pw(G)$, is \mbox{\rm \wonehard} even when $d=2$, $k=\kmin$, and $\avgdeg <  6
\cdot\alpha$.
\end{lemma}
\begin{proof}
The proof is by an {\fpt} reduction from {\mcc}.
We first note that {\mcc} is {\wonehard} even when the number of colors $k$ is restricted to be odd.
Indeed, the problem can be reduced to its special case as follows.
Let $(H,c,k)$ be an instance of {\mcc} where $k$ is even.
Let $(H',c',k+1)$ be an instance of {\mcc}, where $(H',c')$ is obtained by adding a universal vertex $v$ colored $k+1$ to $H$.
It is easy to see that $H$ contains a clique on $k$ vertices with one vertex from each color class if and only if $H'$ contains one. Also, by adding dummy vertices, we can also assume that each chromatic class in the input graph of {\mcc} has at least $t$ vertices, where $t=\frac{\alpha}{\alpha-1}$.

Given an instance $(H,c,k)$ of {\mcc} with $k$ odd, we construct an instance $\inst$ of \mrcf{2}.
Let $V(H) = V_1 \uplus V_2 \uplus \cdots \uplus V_k$, where $V_i$ is the set of vertices of $V(G)$ that are colored $i$, for $i \in [k]$. Notice that $|V(G)\geq t\cdot k$.
Let $W$ be an Eulerian circuit of the complete graph $K_k$, which exists since $k$ is odd. We assume that $V(K_k)=[k]$.
We also assume, for ease of exposition, that $W$ starts with the sequence of vertices $1, 2, \ldots, k$, where every vertex $i$ of this sequence is
considered as the first occurence of $i$ in $W$.
We also assume that the last vertex of $W$ (i.e. the vertex before the first occurence of $1$) is $3$.
Clearly, every $i \in [k]$ appears in $W$ exactly $k' \defined (k-1)/2$ times.
The vertex set of $G$ is the disjoint union of three sets $U, S$, and $T$ where

\begin{figure}
\begin{center}
\scalebox{.9}{\commentfig{\includegraphics[width=\textwidth]{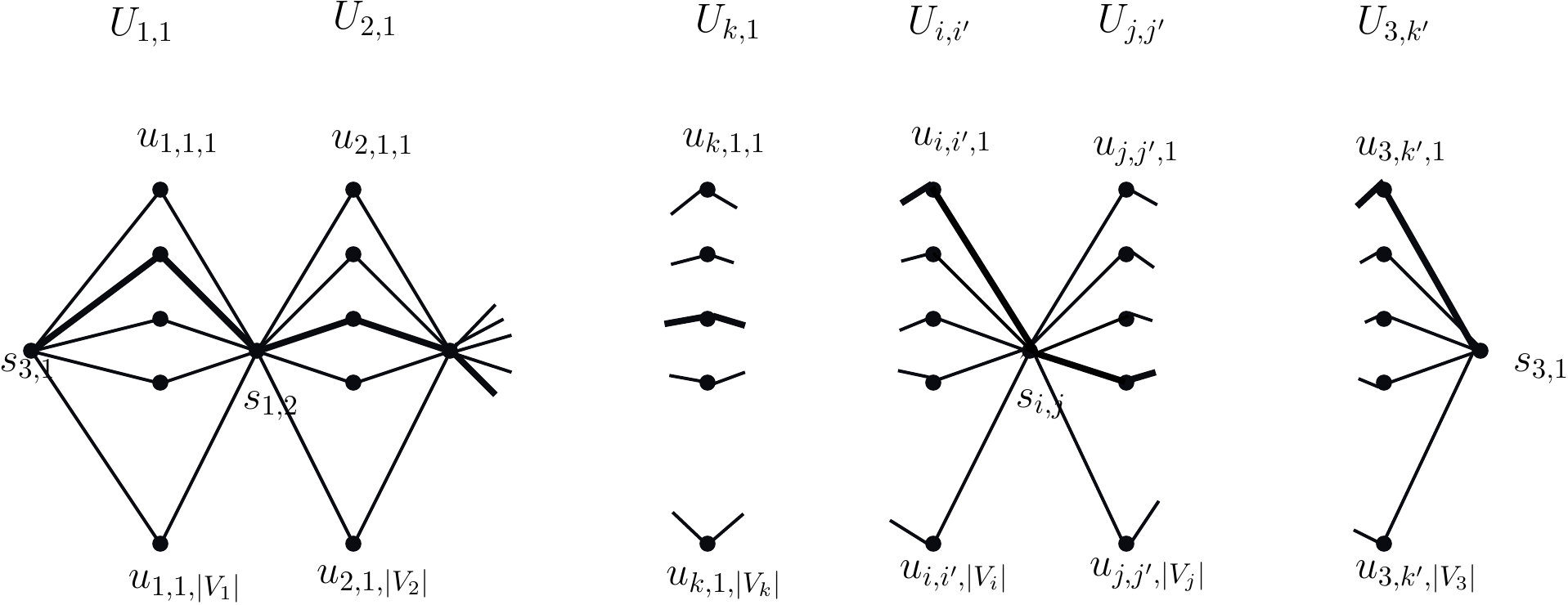}}}
\caption{The $S-U$ edges of $G$. A sample clique cycle appears in bold.}\label{segments}
\end{center}
\end{figure}

\begin{itemize}
\item $U$ consists of $k'$ copies $U_{i,1}, \ldots, U_{i,k'}$ of $V_i$ for every $i \in [k]$, for a total of $k \choose 2$ sets.
For every $i \in [k]$, we number the vertices of $V_i$ from $1$ to $|V_i|$, and
we number the vertices of every copy $U_{i,j}$ accordingly, as
$u_{i,j,1}, \ldots u_{i,j,|V_i|}$.
\item $S$ consists of $k \choose 2$ vertices $s_{i,j}$, one for every arc $i j$ of $W$,
\item $T=T_1 \cup T_2 \cup \cdots \cup T_k$, where every set $T_i$ consists of $|V_i|-1$ vertices
$t_{i,1}, t_{i,2}, \ldots, t_{i,|V_i|-1}$.
\end{itemize}

We proceed with the description of the edge set of $G$, which contains three types of edges depending on their endpoints.
Every edge has one endpoint in $U$ and the other endpoint is in one of $U, S$ or $T$.

\begin{itemize}
\item{$S-U$ edges:}
Let $e=i j$ be an arc of $W$ such that $e$ is incident to the $i'$-th (resp. $j'$-th) occurence of $i$ (resp. $j$) in $W$.
Then $s_{i,j}$ is adjacent to every vertex of $U_{i,i'}$ and to every vertex of $U_{j,j'}$ (see Figure \ref{segments}).

\item{$U-U$ edges:}
The $U-U$ edges form $|V(H)|$ vertex-disjoint paths, one path on $k'$ vertices for every $v \in V(H)$.
The path corresponding to the $\ell$-th vertex of $V_i$ is
is $u_{i,1,\ell} u_{i,2,\ell} \ldots u_{i,k',\ell}$ (see Figure \ref{immature}).

\item{$T-U$ edges:}
For every $i \in [k]$ and every $\ell \in [|V_i|-1]$,
the vertex $t_{i,\ell}$ is adjacent to $u_{i,1,\ell}$, $u_{i,1,\ell+1}$, $u_{i,k',\ell}$ and $u_{i,k',\ell+1}$
(see Figure \ref{immature}).
\end{itemize}

\noindent We prove two properties of  $G$ that are necessary to conclude the lemma.

\begin{claim}\label{pensions}
$\avgdeg   < 6 \alpha$.
\end{claim}

\begin{proof}[Proof of Claim \ref{pensions}]
Clearly, $|U| = k' \cdot |V(H)|$, $|T|=|V(H)|-k$, and $S \neq \emptyset$. Therefore, $|V(G)| > (k'+1) \cdot |V(H)|-k$.

We now proceed with the calculation of $|E(G)|$. Every $U-S$ edge is incident to a vertex of $U$, and every vertex of $U$ has two incident $U-S$ edges.
Therefore, the number of $U-S$ edges is $2 k' \cdot |V(H)|$.
The total number of $U-U$ and $U-T$ edges is $(k'+1) (V(H))$, since for every vertex $v \in V(H)$ we have a cycle that contains $k'+1$ vertices, namely all the copies of $v$ and a $t$ vertex. Thus, $|E(G)|=(3k'+1) \cdot |V(H)|$. As $|V(G)\geq t\cdot k\Rightarrow |V(G)|/(V(G)-k)\leq t/(t-1)=\alpha$, we conclude that $\avgdeg$ is
\[
\frac{2 |E(G)|}{|V(G)|} < \frac{2(3k'+1) \cdot |V(H)|}{(k'+1) \cdot (|V(H)|-k)}< 6 \cdot \frac{|V(H)|}{|V(H)|-k} < 6\alpha.
\]
\end{proof}

\begin{claim}\label{presents}
$\pw(G) \leq {k \choose 2} + 3$.
\end{claim}

\begin{proof}[Proof of Claim \ref{presents}]
Clearly, $\pw(G) \leq \pw(G \sm S) + |S| = \pw(G \sm S) + {k \choose 2}$.
We observe that the edge set of $G \sm S$ consists of $U - U$ and $T-U$ edges that form connected components $T_i \cup \cup_{j'=1}^{k'} U_{i,j'}$ as depicted in Figure \ref{immature}.
Every such component has a path decomposition of $k' \cdot |V_i|$ bags with at most $4$ vertices in a bag, as depicted in Figure \ref{teaching}.
Therefore, $\pw(G \sm S) \leq 3$ and the claim follows.
\end{proof}

\begin{figure}
\begin{center}
\commentfig{\includegraphics[scale=0.750]{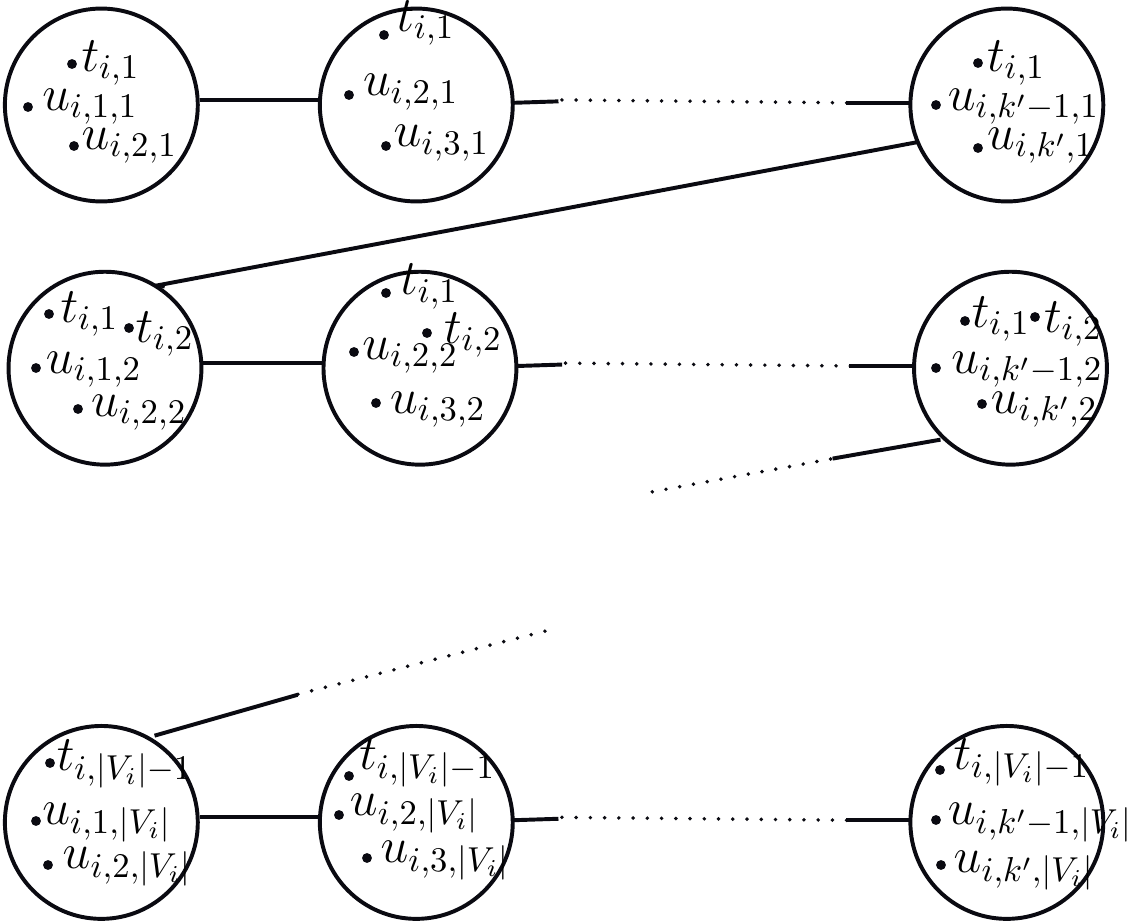}}
\caption{A path decomposition of a connected component of $G \sm S$.} \label{teaching}
\end{center}
\end{figure}

We proceed with the description of the coloring function $\chi$ and the traversal cost function $c$.
The color set is $V(H) \cup \{ white \}$. In other words, the vertices of $H$ corresponds to colors in the constructed instance $H'$; that is, there is a color in $H'$ corresponding to each vertex in $H$. All the $U-U$ and $T-U$ edges are colored white, and the $S-U$ edges are colored upon their endpoint in $U$ as follows.
For every $S-U$ edge $e$ with endpoint $u_{i,j,\ell}$, we define $\chi(e)=v_{i,\ell}$. The traversal cost $c(v,white)$ is $1$ for every $v \in V(H)$ and $c(white,white)=0$.
Finally, for every $v,v' \in V(H)$
\[
c(v,v')= \left\{ \begin{array}{ll}
 0 & \textrm{if~} v v' \in E(H) \textrm{~or~} v = v'\\
 1 & \textrm{otherwise.}
 \end{array} \right.
\]
Finally, we set $k=0$.
This completes the construction of $\inst$, which can be clearly performed in polynomial time.

\begin{figure}[ht]
\begin{center}
\commentfig{\includegraphics[scale=0.82]{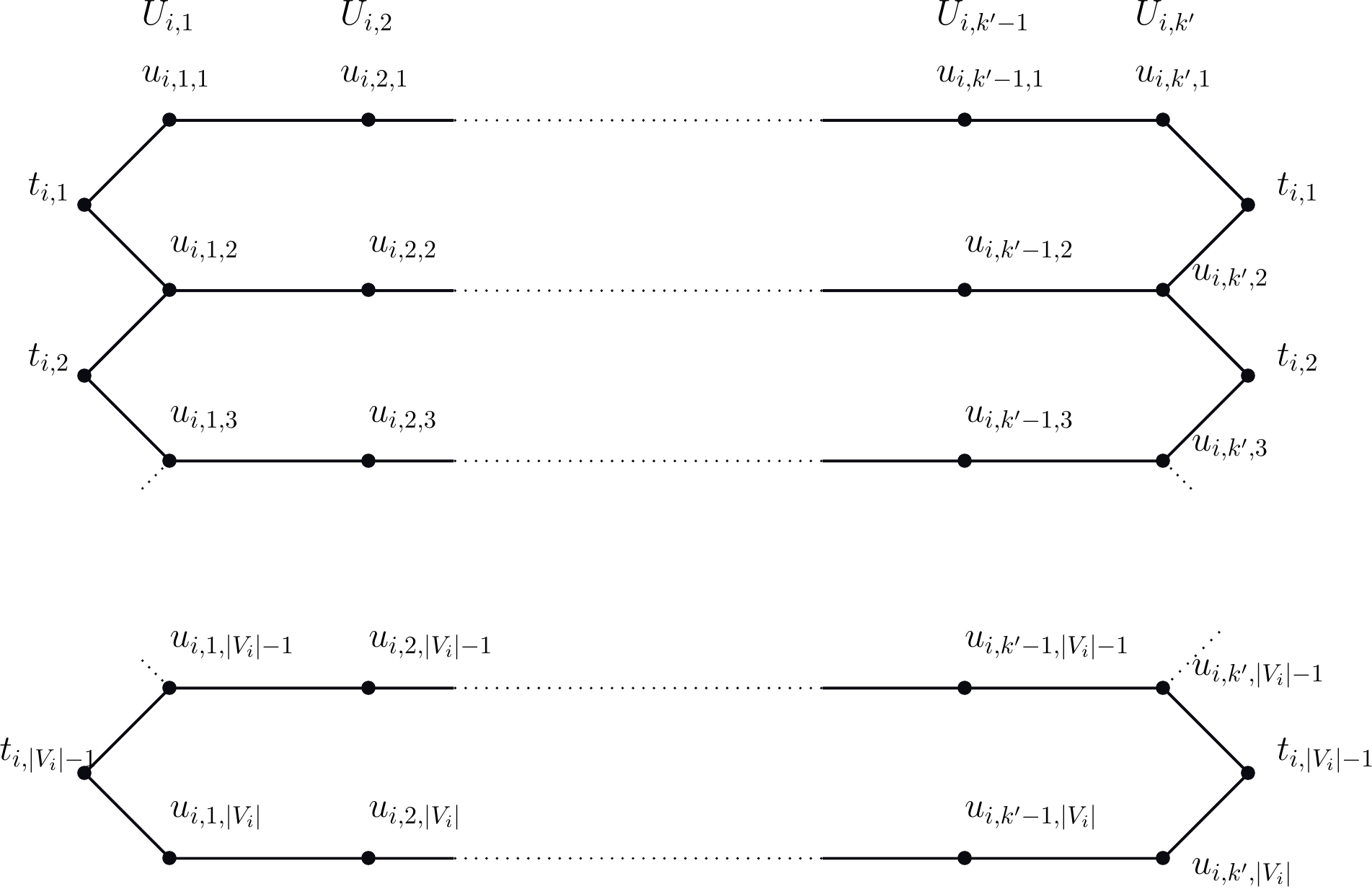}}
\caption{The $U-U$ and $T-U$ edges of $G$.} \label{immature}
\end{center}
\end{figure}

We now prove that $H$ has a $k$-clique with one vertex from every color class if and only if $G$ has a $2$-factor $F$ with $\rc(F)=0$.
Assume that $H$ has a clique $K$ with exactly one vertex from every color, and suppose without loss of generality that $K$ consists of the first vertices of each color class.
Let $F$ be the $2$-factor of $G$ consisting of the following cycles.
\begin{itemize}
\item {}
The clique cycle $G[C_K]$, where $C_K=S \cup \{ u_{i,j,1} \mid i \in [k], j \in [k'] \}$.
\item {}
The vertex cycles, one per every vertex $v \in V(H) \sm K$.
For every $i \in [k]$ and every $\ell \in [2,|V_i|]$, the vertex cycle corresponding to the vertex $v_{i,\ell} \in V(H) \sm K$ is $G[C_{i,\ell}]$, where $ C_{i,\ell} = \{ t_{i,\ell-1}, u_{i,j,\ell} \mid j \in [k'] \}$.
\end{itemize}
It is easy to see that every vertex is in exactly one of these cycles.
Furthermore, the edges of the vertex cycles are all white, incurring a reload cost of zero.
It remains to show that that the edges of the clique cycle also incur a zero cost.
Indeed, the two edges incident to a vertex $u_{i,j,\ell} \in U$ incur a cost $c(v_{i,\ell},v_{i,\ell})=0$,
and the two edges incident to a vertex $s_{i,j}$ incur a zero cost since the adjacent vertices are $u_{i,i',1}$ and $u_{j,j',1}$ that correspond to the vertices $v_{i,1}$ and $v_{j,1}$ of $K$, which are adjacent in $H$.

Conversely, suppose that $H$ contains a 2-factor $F$ with $\rc(F)=0$.
Every vertex $s_{i,j}$ is adjacent to two distinct vertices of $U$ that correspond to two distinct vertices $v,v'$ of $V(H)$.
Furthermore, since the cost at vertex $s_{i,j}$ is zero, we conclude that $v$ and $v'$ are adjacent in $H$.
In particular, $v$ and $v'$ are in different color classes.
Therefore, $s_{i,j}$ is adjacent to one vertex from each of $U_{i,i'}$ and $U_{j,j'}$.
The second edge incident to these vertices in $F$ is not white, since otherwise it incurs a positive traversal cost.
Thus, the other edge is also a $U-S$ edge.
We conclude that the $k \choose 2$ vertices $S$ are all in one cycle $C_K$ that also contains $k \choose 2$ vertices from $U$, one from every set $U_{i,i'}$.
Furthermore, two consecutive vertices of $U$ in $C_K$ correspond to two adjacent vertices of $H$.

All the remaining cycles must be in $G \sm S$, which consists of $k$ connected components $T_i \cup \cup_{i'=1}^{k'} U_{i,i'}$, for every $i \in [k]$.
The result is now apparent from Figure \ref{immature} that depicts this connected component.
If a vertex $u_{i,i',\ell}$ is in one of the remaining cycles, then all the vertices $u_{i,i'',\ell}$ for $i'' \in [k']$ are in the same cycle as $u_{i,i',\ell}$.
Therefore, if a vertex $u_{i,i',\ell}$ is in $C_K$, then all the vertices $u_{i,i'',\ell}$ for $i'' \in [k']$ are in $C_K$.
Recall that $C_K$ contains one vertex $u_{i,i',\ell}$ from every set $U_{i,i'}$.
We conclude that these vertices correspond to the same vertex $v_{i,\ell}$ of $V_i$.
We recall that consecutive $U$-vertices of $C_K$ correspond to two adjacent vertices of $H$.
Therefore, the vertices $U \cap C_K$ correspond to the vertices of a clique of $H$.
\end{proof}

Given a graph $G$ on $n$-vertices where $n$ is even, 
a {\em pairing collection} of $G$ is a collection $M$
of pairs of vertices in $G$ that forms a partition of $V(G)$.
We denote $G+M=(V(G),E(G)\cup M)$. Notice that $G+M$ can be a multi-graph as the pairings in $M$ might already be edges of $G$.
We prove the following lemma that will prove useful in our proof.

\begin{lemma}\label{daylight}
Let $G$ be a $n$-vertex graph where $n$ is even where $\pw(G)\leq k$ for some $k\geq 1$. Then $G$ has a pairing collection $M$ such that $\pw(G+M)\leq k+1$.
\end{lemma}
%
%
%

\begin{proof}
We say that a vertex $v$ of a graph is {\em $k$-simplicial}
if its neighbourhood induces a complete graph on $k$ vertices. 
{\em Proper $k$-paths} are recursively defined as follows:
A complete graph on $k+1$ vertices is a proper $k$-path. 
A graph $H$ on more than $k+1$ vertices is a proper $k$-path
if it contains a $k$-simplicial vertex $v$ such that the removal of $v$
results to a proper $k$-path with a simplicial vertex that belongs in the neighbourhood  of $v$ in $H$. Notice that this recursive 
definition gives rice to sequence of simplicial vertices where every two successive vertices are adjacent and such that all vertices of $G$ appear in this sequence. Therefore every proper $k$-tree $H$
has a hamiltonian path, which, in case $H$ has an even number of vertices, implies that $H$ contains a perfect matching.
It was proved in~\cite{Tak95}  (see also~\cite{KS96siamjomput}) that every graph $G$ of pathwidth $\leq k$ is a spanning subgraph of a proper $(k+1)$-path $H$. As a perfect 
matching in $H$ is a pairing collection in $G$, the lemma follows.
\end{proof}

\begin{theorem}\label{dispense}
For every $r \in \mathbb{N}_{\geq 2}$, \rmrcf, parameterized by $\pw(G)$, is {\rm \wonehard} even when $d=2$, $k=\kmin$ and $\avgdeg$ is less than
$r+\left\{ \begin{array}{ll}
4 & \textrm{~if~} r=2 \\ 4/3 & \textrm{otherwise}.
\end{array}\right.$
\end{theorem}

\begin{proof}

For $r=2$ the Theorem is equivalent to Lemma \ref{granting}.
For $r>2$ we present an $\fpt$ reduction from the base case, i.e. $r=2$.
using the same technique as in the proof of Theorem \ref{consiste}.
Namely, given an instance $\inst$ of \mrcf{2} with $d=2$ and $k=\kmin=0$ where $\avgdeg < 6$,
we construct an instance $\instprime$ of {\rmrcf} where $d=2$ and the average degree of $G'$ is less than $r+4/3$
by adding to $G$ a pair collection, and replacing every new edge by the gadget $Q_r$.
By the last part of the proof of Theorem \ref{consiste}, $G'$ has an $r$-factor of zero cost if and only if $G$ has a $2$-factor of zero cost.

\begin{claim}\label{connects}
The average degree of $G'$ is at most $r+4/3$.
\end{claim}

\begin{proof}[Proof of Claim \ref{connects}]
Let $n=|V(G)|$. Then $|E(G)| < 3n$.
On the other hand, the number of edges added by every $Q_r$ is $\frac{r(r+1)}{2}+r-2$.
Summarizing, we have
\begin{eqnarray*}
|E(G')| & <  & 3n +  \frac{n}{2} \left(\frac{r(r+1)}{2}+r-2 \right)\\
|V(G')| & =  & n + \frac{n}{2} (r+1).
\end{eqnarray*}
We conclude that the average degree of $G'$ is less than
\[
\frac{6 +  \left(\frac{r(r+1)}{2}+r-2 \right)}{1 + \frac{r+1}{2}} = \frac{r^2+3r+8}{r+3}=r+\frac{8}{r+3} \leq r+\frac{4}{3}.
\]
\end{proof}

\begin{claim}\label{handling}
$\pw(G') \leq (\pw(G)+1)\cdot \frac{r+1}{r-1}-1$.
\end{claim}

%
\begin{proof}[Proof of Claim \ref{handling}]
According to Lemma~\ref{daylight},  there is a pairing collection $M$ of $G$ such that $\pw(G + M) \leq \pw(G)+1$.

The operation of replacing each pair of $M$ by a $Q_r$ can be reflected in the path decomposition as follows. Each bag has $k+1$ vertices 
and among the edges with endpoints in this bag, at most $(k+1)/2$
of them can be pairs that belong in $M$. This means that after the insertion of 
at most $(k+1)/2$ copies of $Q_{r}$ in each bag the new path decomposition
will have width at most $k+1+(k+1)/2\cdot (r+1)=(\pw(G)+1)\cdot \frac{r+1}{r-1}$.
\end{proof}
Finally, we note that $r$ is a constant of the problem. 
Therefore, the function proven in Claim \ref{handling}, is a function of $\pw(G)$ as required.
\end{proof}

\newcommand{\alg}{{\sc DynProg}}

We conclude with the following algorithmic result.

\begin{theorem}\label{softened}
For every $r\in\mathbb{N}_{\geq 2}$, {\sc $r$-MRCF}, parameterized by $\min \{ q, \Delta(G) \}$ and
$\tw(G)$ is in {\fpt}. 
Specifically, it can be solved by an algorithm that runs in time 
$\OO^* \left( {{\min \{ q, \Delta(G) \} +r} \choose r}^{2(\tw(G)+1)} \right)$.
\end{theorem}

%
%

\begin{proof}
We assume that the graph $G$ is given together with a nice pair $(\DD, \GG)$ where $\DD=(T, \XX, r)$.
Indeed, by using for instance the algorithm of Bodlaender \emph{et al}.~\cite{BodlaenderDDFLP16}, we can compute in time $2^{\OO(\tw(G))} \cdot n$ a tree decomposition of $G$ of width at most $5 \tw(G)$.
Note that this running time is clearly dominated by the running time stated in Theorem~\ref{softened}.
Recall also that, by~\cite{CyganNPPRW11,Klo94}, given a tree decomposition, it can be transformed in polynomial time
into a nice pair.

We first introduce some notation.
Let $t \in T$, $v \in X_t$, and $F$ an $r$-factor of $G$.
The graph $F_t = F[V(G_t)]$ is a partial $r$-factor of $G$.
We encode the edges $E_{G_t}(v)$ incident to $v$ in $G_t$ as a vector $\vect{E}_{t,v}$ of non-negative integers indexed by the colors of the instance.
The value of the entry $\lambda$ of $\vect{E}_{t,v}$ is the number of edges of $E_{G_t}(v)$ colored $\lambda$.
Clearly, the sum of the entries of this vector, denoted by $\norm{\vect{E}_{t,v}}$ is $|E_{G_t}(v)|$.
The encoding of a single edge $e$ colored $\chi(e)=\lambda$, incident to $v$ is the vector $\vect{u}_\lambda$ that has a $1$ in entry $\lambda$ and zero elsewhere.
We now define the set of all vectors encoding a subset of $|E_{G_t}(v)$ consisting of at most $r$ edges.
In other words, these are the sets of edges incident to $v$ in $G_t$ that can be part of an $r$-factor of $G$:
\[
\EE_{t,v} \defined \{ \vect{E} \mid \vect{E} \leq  \vect{E}_{t,v}, \norm {\vect{E}} \leq r \}.
\]
Note that $|\EE_{t,v}| \leq {{q+r} \choose r}$.
Let $\vect{E}, \vect{E}' \in \EE_{t,v}$ encode two disjoint sets of edges incident to $v$.
Each of $\vect{E}, \vect{E}'$ has an associated reload cost on $v$.
The reload cost of their union $\vect{E} + \vect{E}'$  is the sum of their individual reload costs, plus the reload costs incurred by the traversals between the two sets, which is
$
c_v(\vect{E},\vect{E}') \defined \sum_{\lambda} \sum_{\lambda'} c(\lambda, \lambda') \cdot \vect{E}_\lambda \cdot \vect{E}_{\lambda'}.
$

\alglanguage{pseudocode}

\begin{algorithm}[H]
\caption{\alg}\label{begrudge}
\begin{algorithmic}[1]
\Require{An instance $\inst$ of {\rmrcf}}
\Require{A nice pair $\DD, \GG$ where $\DD=(T,r,\XX)$ is a tree decomposition of $G$.}
\Statex
\For {$t \in T$ (using a bottom-up traversal)}
    \State $\TT_k(\vect{E_t}) \gets \infty$ for every $\vect{E_t} \in \EE_t$.
    \If{$t=r$}
        \State \Return $\TT_r(()) \leq k$.
               \Comment $T_r=\es, \EE_r=\set{()}$
    \ElsIf{$t$ is a leaf node}
        \State $T_t(())=0$.
    \ElsIf{$t$ is a join node with children $t',t''$}
        \For {$\vect{E_{t'}} \in \EE_{t'}$}
            \For {$\vect{E_{t''}} \in \EE_{t''}$}
                \If {$\vect{E_{t'}}+\vect{E_{t''}} \in \EE_t$}
                    \State $cost \gets \TT_{t'}(\vect{E_{t'}}) + \TT_{t''}(\vect{E_{t''}}) + c_{X_t}(\vect{E_{t'}}, \vect{E_{t''}})$.
                    \If {$cost < \TT_t(\vect{E_{t'}}+\vect{E_{t''}})$} $\TT_t(\vect{E_{t'}}+\vect{E_{t''}}) \gets cost$
                    \EndIf
                \EndIf
            \EndFor
        \EndFor
    \ElsIf{$t$ introduces the vertex $v$ and has child $t'$}
        \For {$\vect{E_{t'}} \in \EE_{t'}$}
            \State $\TT_t( (\zero, \vect{E_{t'}}) ) = \TT_{t'}(\vect{E_{t'}})$.
            \Comment w.l.o.g. $v$ is the first vertex of $X_t$.
        \EndFor
    \ElsIf{$t$ introduces the edge $e=uv$ and has child $t'$}
    \Comment $\EE_{t'} \subseteq \EE_t$.
        \For {$\vect{E_t} \in \EE_{t'}$}
        \Comment Partial $r$-factors that do not contain $e$
            \State $\TT_t(\vect{E_t}) \gets \TT_{t'}(\vect{E_t})$
        \EndFor
        \For {$\vect{E_t} \in \EE_{t} \sm \EE_{t'}$}
        \Comment Partial $r$-factors that contain $e$
            \State $\vect{E_{t'}} \gets \vect{E_t} - (\vect{u_{\chi(e)}},\vect{u_{\chi(e)}},\zero,\ldots,\zero)$
            \State $\TT_t(\vect{E_t}) \gets \TT_{t'}(\vect{E_{t'}}) + c_u(\vect{E_{t'}},\vect{u_{\chi(e)}}) + c_v(\vect{E_{t'}},\vect{u_{\chi(e)}})$
        \EndFor
    \Else
        \Comment $t$ is a forget node that forgets $v$ and has child $t'$
        \For {$\vect{E_{t'}} = (\vect{e_v},\vect{E_t}) \in \EE_{t'}$}
            \If {$\norm{\vect{e_v}}=r$ and $\TT_{t'}(\vect{E_{t'}}) < \TT_t(\vect{E_t})$}
                \State $\TT_t(\vect{E_t}) \gets \TT_{t'}(\vect{E_{t'}})$
            \EndIf
        \EndFor
    \EndIf
\EndFor

\end{algorithmic}
\end{algorithm}

\noindent We extend this definition to any subset $X$ of vertices as $c_X(\vect{E},\vect{E}') \defined \sum_{x \in X} c_x(\vect{E},\vect{E}')$.

$\EE_t$ is the cartesian product $\times_{v \in X_t} \EE_{t,v}$, i.e. every element of $\EE_t$ is a vector of $|X_t|$ vectors that encodes a set of edges incident to $X_t$ that is possibly part of an $r$-factor of $G$.
Clearly, $|\EE_t| \leq {{q+r}  \choose r}^{|X_t|} \leq { {q+r} \choose {r}}^{w+1}$.
For every $t \in T$ we maintain a function $\TT_t : \EE_t \to \mathbb{N}_{\geq 0}$ that returns for any given $\vect{E_t} \in \EE_t$, the minimum cost of a partial $r$-factor whose edges incident to $X_t$ are encoded by $\vect{E_t}$.

The dynamic programming algorithm whose pseudo code is given in Algorithm \ref{begrudge} performs a bottom-up traversal of $T$ and updates the tables $\TT_t$ at every node, and finally returns the result by inspecting the value $\TT_r(())$ (recall that $\EE_r = \set{()}$.

As for the running time of the algorithm, we observe that the dominant running time for a node is obtained in the case of a join node,
in which case the running time is $\OO(|\EE_{t'}| \cdot |\EE_{t''}|) = \OO^* \left( {{q+r} \choose r}^{2(\tw(G)+1)} \right)$, ignoring polynomial factors.
Alternatively, we can encode the vectors of $\EE_{t,v}$ using a binary vector of length $\Delta(G)$.
In this case $\abs{\EE_t} = {\Delta(G)+r \choose r}$, implying the running time in the statement of the Theorem.
\end{proof}

%

\end{document}